\documentclass[lettersize,journal]{IEEEtran}

\usepackage[utf8]{inputenc}
\usepackage[T1]{fontenc}
\usepackage{color}
\usepackage{cite} 
\usepackage{citesort} 
\usepackage{array}
\usepackage{graphicx}
\usepackage{algorithm}
\usepackage{algorithmic}
\usepackage{tikz}
\usetikzlibrary{arrows.meta, shapes.geometric, calc, positioning, shadows}
\usepackage{amsmath, amssymb, amsfonts}
\usepackage{amsmath,amssymb,amsthm} 
\newtheorem{proposition}{Proposition}
\newtheorem{remark}{Remark}

\usetikzlibrary{calc,arrows.meta,positioning,backgrounds,fit}
\usepackage{subcaption}

\begin{document}

\title{Secrecy Outage Analysis over Correlated Composite Generalized-Gamma Fading Channels}

\author{Muhammad Adil, Pranavkumar Pathak, and Salman A. Alqahtani
\thanks{}
\thanks{Muhammad Adil is with the Department of Electronics Engineering,
University of Rome Tor Vergata, 00133 Rome, Italy
(e-mail: muhammad.adil@uniroma2.it).}
\thanks{Pranavkumar Pathak is with the School of Continuing Studies, McGill University, Montreal, QC H3A 2M7, Canada (e-mail: pranavpp@gmail.com).}
\thanks{Salman A. AlQahtani is with the New Emerging Technologies and 5G Network and Beyond Research Chair, Department of Computer Engineering, College of Computer and Information Sciences, King Saud University, Riyadh, Saudi Arabia (email: salmanq@ksu.edu.sa).}
}

\markboth{}%
{Shell \MakeLowercase{\textit{et al.}}: A Sample Article Using IEEEtran.cls for IEEE Journals}

\maketitle

\begin{abstract}
This paper investigates physical-layer security (PLS) over correlated composite generalized-Gamma (GG)/GG fading channels, where both shadowing and small-scale fading follow GG distributions. Using Mellin transforms and Fox-$H$ functions, closed-form expressions are derived for the single-link probability density function (PDF), joint distribution, survival function, and zero-rate secrecy outage probability (SOP)/probability of non-zero secrecy capacity (PNZSC). The general-rate SOP is expressed as an exact double series with one residual one-dimensional integral per term. The model includes the Nakagami-$m$/GG and Nakagami-$m$/Gamma channels as special cases. Numerical results validate the analysis and demonstrate the impact of the fading parameters on secrecy performance.
\end{abstract}


\section{Introduction}

\IEEEPARstart{P}{hysical}-layer security (PLS) has emerged as a promising complement to conventional cryptographic techniques by exploiting the randomness of wireless channels to enhance confidentiality in sixth-generation (6G) wireless communication networks~\cite{11036803}. Following Wyner's wiretap-channel formulation \cite{6772207}, secure transmission is possible when the legitimate receiver experiences a channel advantage over the eavesdropper. However, this advantage is not always guaranteed in practical wireless environments, especially when the legitimate and eavesdropping links are affected by similar propagation conditions, common scatterers, or spatially correlated shadowing. As a result, accurate secrecy performance analysis under correlated fading and shadowing is essential for assessing the reliability of PLS techniques in realistic wireless networks.

Several works have investigated secrecy capacity, secrecy outage probability (SOP), and the probability of non-zero secrecy capacity (PNZSC) over classical fading channels. Correlated Rayleigh and Nakagami-\(m\) channels have been studied to quantify the secrecy degradation caused by spatial correlation \cite{8060499}. Log-normal shadowing models have also been considered for scenarios where large-scale fading dominates the secrecy behavior. More recently, composite fading models have received increasing attention because they jointly capture small-scale multipath fading and large-scale shadowing effects \cite{9239955,9768319}. In particular, the Nakagami-\(m\)/Gamma model provides a mathematically tractable framework for composite fading, while generalized-Gamma shadowing further improves modeling flexibility by including several conventional distributions as special cases \cite{8060499}.

Although these models provide valuable insight, most existing secrecy analyses assume that at least one propagation component remains restricted to a conventional Gamma-based form. For example, small-scale fading is often modeled through Nakagami-\(m\), whose received power is Gamma distributed, while generalized-Gamma (GG) modeling is applied only to the shadowing component. This assumption limits the flexibility of the channel model in environments where both the multipath and shadowing processes deviate from standard Gamma behavior. Such deviations may arise in dense urban deployments, indoor propagation, blockage-prone links, and heterogeneous wireless networks where the fading severity and tail behavior cannot be accurately represented by a single-parameter Nakagami-\(m\) model \cite{8917797}.

To address this limitation, we develop a secrecy outage framework for correlated composite generalized-Gamma/generalized-Gamma (GG/GG) fading channels. The legitimate and eavesdropping links undergo independent GG small-scale fading and correlated shadowing. Using Mellin transforms and Fox-(H) functions, we derive closed-form expressions for the single-link PDF, joint distribution, survival function, and zero-rate SOP/PNZSC, together with an exact double-series representation involving one residual one-dimensional integral per term for the general-rate SOP. The proposed framework unifies and extends several existing secrecy models. By setting the GG small-scale exponent to unity, the model reduces to the Nakagami-\(m\)/GG secrecy model. By further setting the shadowing exponent to unity, it reduces to the classical Nakagami-\(m\)/Gamma composite fading model \cite{8060499}. Therefore, the derived expressions provide both a generalized analytical tool and a consistency check with known results. Numerical and Monte Carlo results are finally used to validate the analysis and to demonstrate the impact of GG fading parameters and shadowing correlation on secrecy outage performance.

The main contributions of this paper are summarized as follows:
\begin{itemize}
    \item A correlated composite GG/GG secrecy model is developed, encompassing the Nakagami-\(m\)/GG and Nakagami-\(m\)/Gamma models as special cases, with closed-form expressions for the single-link PDF, joint distribution, survival function, and zero-rate SOP/PNZSC derived using Mellin transforms and Fox-\(H\) functions.

    \item An exact double-series representation with a one-dimensional integral per term is derived and validated through numerical and Monte Carlo results, which illustrate the effects of the GG fading parameters and shadowing correlation.
\end{itemize}

The rest of this paper is organized as follows. Section~\ref{sec:sysmo} presents the system and channel model. Sections~\ref{sec:slpdf} and~\ref{sec:jointpdf} derive the single-link and joint GG/GG fading distributions, respectively. Section~\ref{sec:secrecy_performance} develops the secrecy outage and PNZSC analyses. Section~\ref{sec:numerical_results} presents the numerical results and discussion, while Section~\ref{sec:conclusion} concludes the paper.

\section{System and Channel Model}
\label{sec:sysmo}

We consider a wiretap system composed of a legitimate transmitter Alice denoted by \(A\), a legitimate receiver denoted by \(B\), and an eavesdropper denoted by \(E\). For \(\ell\in\{\mathrm B,\mathrm E\}\), respectively, the received signal at node \(\ell\) is given by \cite{8060499}
\begin{equation}
    y_\ell=\sqrt{p}\,h_\ell s+n_\ell,
    \label{eq:received_signals}
\end{equation}
where \(p\) denotes the transmit power, \(s\) is the transmitted information symbol satisfying \(\mathbb{E}\{|s|^2\}=1\), \(h_\ell\) is the complex channel coefficient of link \(\ell\), and \(n_\ell\) denotes an additive white Gaussian noise (AWGN) sample with variance \(\sigma_\ell^2\).

The channel amplitude of link \(\ell\) is modeled as a composite fading random variable given by
\begin{equation}
    |h_\ell|=\sqrt{B_\ell}\,W_\ell,
    \label{eq:channel_amplitudes}
\end{equation}
where \(B_\ell\) represents the large-scale shadowing power component, while \(W_\ell\) denotes the small-scale fading amplitude.

The shadowing component is modeled as
\begin{equation}
    B_\ell = Y_\ell^{1/\beta_\ell}, \qquad \beta_\ell>0,
    \label{eq:shadowing_model}
\end{equation}
where \((Y_{\mathrm B},Y_{\mathrm E})\) is a correlated Gamma random pair with marginal shape parameters \(k_{\mathrm B}\) and \(k_{\mathrm E}\), scale parameters \(\theta_{\mathrm B}\) and \(\theta_{\mathrm E}\), and correlation parameter \(\rho\in[0,1)\), which characterizes the dependence of the underlying Gamma pair \cite{5281755}. Since \(Y_\ell\sim\mathrm{Gamma}(k_\ell,\theta_\ell)\), the shadowing power component \(B_\ell=Y_\ell^{1/\beta_\ell}\) follows a generalized-Gamma distribution characterized by the shape parameter \(k_\ell\) and the power exponent \(\beta_\ell\). The parameter \(\beta_\ell\) controls the GG shadowing behavior.

Unlike the conventional Nakagami-\(m\) assumption, where the small-scale fading power is Gamma distributed, i.e., \(W_\ell^2\sim\mathrm{Gamma}(m_\ell,\Omega_\ell/m_\ell)\), with \(\Omega_\ell\triangleq\mathbb{E}\{W_\ell^2\}\), we model the small-scale fading power as a GG random variable. Specifically,
\begin{equation}
    W_\ell^2 = V_\ell^{1/\kappa_\ell},\qquad
    V_\ell\sim \mathrm{Gamma}(m_\ell,\tau_\ell),\qquad
    \kappa_\ell>0,
    \label{eq:smallscaleGG}
\end{equation}
where \(m_\ell\) is the small-scale fading shape parameter, \(\tau_\ell\) is the scale parameter of the auxiliary Gamma random variable \(V_\ell\), and \(\kappa_\ell\) controls the GG small-scale fading behavior. The random variables \(V_{\mathrm B}\) and \(V_{\mathrm E}\) are assumed to be mutually independent and also independent of the correlated shadowing pair \((Y_{\mathrm B},Y_{\mathrm E})\). When \(\kappa_\ell=1\), the model in \eqref{eq:smallscaleGG} reduces to Nakagami-\(m\) small-scale fading. Furthermore, setting \(\kappa_{\mathrm B}=\kappa_{\mathrm E}=1\) and \(\beta_{\mathrm B}=\beta_{\mathrm E}=1\) recovers the classical Nakagami-\(m\)/Gamma composite fading model \cite{8060499}.

The scale parameter \(\tau_\ell\) is chosen to satisfy the prescribed average small-scale fading power \(\Omega_\ell\). Therefore, with \(\Gamma(\cdot)\) denoting the Gamma function,
\begin{equation}
\tau_\ell =
\left(
\Omega_\ell
\frac{\Gamma(m_\ell)}
{\Gamma(m_\ell+1/\kappa_\ell)}
\right)^{\kappa_\ell}.
\label{eq:tautomean}
\end{equation}

Let \(c_\ell\triangleq p/\sigma_\ell^2\). The average received SNR
of link \(\ell\) is
\begin{equation}
\bar{\gamma}_\ell
=
c_\ell\Omega_\ell\,
\theta_\ell^{1/\beta_\ell}
\frac{\Gamma\!\left(k_\ell+1/\beta_\ell\right)}
{\Gamma(k_\ell)}.
\label{eq:average_snr}
\end{equation}
Hence, for a prescribed average SNR \(\bar{\gamma}_\ell\), the
shadowing scale parameter is selected as
\begin{equation}
\theta_\ell
=
\left[
\frac{\bar{\gamma}_\ell\Gamma(k_\ell)}
{c_\ell\Omega_\ell
\Gamma\!\left(k_\ell+1/\beta_\ell\right)}
\right]^{\beta_\ell}.
\label{eq:theta_from_snr}
\end{equation}

For analytical convenience, we henceforth use the normalized setting
\(c_\ell=\Omega_\ell=1\). For every prescribed average SNR
\(\bar{\gamma}_\ell\), the scale parameters \(\theta_\ell\) and
\(\tau_\ell\) are determined from \eqref{eq:theta_from_snr} and
\eqref{eq:tautomean}, respectively. Consequently,
\begin{equation}
\begin{aligned}
\gamma_{\mathrm B}
&=B_{\mathrm B}W_{\mathrm B}^2
=Y_{\mathrm B}^{1/\beta_{\mathrm B}}
V_{\mathrm B}^{1/\kappa_{\mathrm B}},\\
\gamma_{\mathrm E}
&=B_{\mathrm E}W_{\mathrm E}^2
=Y_{\mathrm E}^{1/\beta_{\mathrm E}}
V_{\mathrm E}^{1/\kappa_{\mathrm E}}.
\end{aligned}
\label{eq:gammadef}
\end{equation}
with the SOP and PNZSC, for a target secrecy rate \(r\geq0\), defined as
\begin{equation}
P_o(r)\triangleq
\Pr\!\left[
\gamma_{\mathrm B}
\leq 2^r(1+\gamma_{\mathrm E})-1
\right],
\qquad
\mathrm{PNZSC}\triangleq 1-P_o(0).
\label{eq:SOP_PNZSC_def}
\end{equation}

\section{Single-Link Composite GG/GG PDF}
\label{sec:slpdf}

\subsection{Mellin Transform}

For a single link, let \(Z=BU\), where \(B=Y^{1/\beta}\), \(Y\sim\mathrm{Gamma}(k,\theta)\), \(U=V^{1/\kappa}\), and \(V\sim\mathrm{Gamma}(m,\tau)\), with \(Y\) and \(V\) independent. For any positive random variable \(X\), define its Mellin transform as \(\mathcal M_X(\zeta)\triangleq\mathbb E[X^{\zeta-1}]\), where \(\zeta\in\mathbb C\) is the Mellin-transform variable. The corresponding Mellin transforms are
\begin{equation}
\begin{aligned}
\mathcal M_Y(\zeta)
&=\theta^{\zeta-1}\frac{\Gamma(k+\zeta-1)}{\Gamma(k)},\\
\mathcal M_B(\zeta)
&=\theta^{\frac{\zeta-1}{\beta}}
\frac{\Gamma\!\left(k+\frac{\zeta-1}{\beta}\right)}{\Gamma(k)}.
\end{aligned}
\label{eq:mellin_Y_B}
\end{equation}
Similarly,
\begin{equation}
\begin{aligned}
\mathcal M_V(\zeta)
&=\tau^{\zeta-1}\frac{\Gamma(m+\zeta-1)}{\Gamma(m)},\\
\mathcal M_U(\zeta)
&=\tau^{\frac{\zeta-1}{\kappa}}
\frac{\Gamma\!\left(m+\frac{\zeta-1}{\kappa}\right)}{\Gamma(m)}.
\end{aligned}
\label{eq:mellin_V_U}
\end{equation}
Since \(B\) and \(U\) are independent,
\begin{equation}
\mathcal M_Z(\zeta)
=D^{\zeta-1}
\frac{\Gamma\!\left(k+\frac{\zeta-1}{\beta}\right)
\Gamma\!\left(m+\frac{\zeta-1}{\kappa}\right)}
{\Gamma(k)\Gamma(m)},
\label{eq:MZ}
\end{equation}
where \(D\triangleq\theta^{1/\beta}\tau^{1/\kappa}\).

The expression in \eqref{eq:MZ} reveals a symmetric dependence on the shadowing parameters \((k,\beta)\) and the small-scale fading parameters \((m,\kappa)\), which will be exploited in the subsequent Fox-\(H\) representation.

\subsection{Fox-\(H\) Function Representation}
\begin{proposition}\label{prop:1prime}
For \(z>0\), the PDF of \(Z\) is
\begin{equation}
\begin{aligned}
\kappa(z;k,m,\beta,\kappa_{\mathrm e},D)
&=\frac{1}{D\Gamma(k)\Gamma(m)}
H^{2,0}_{0,2}\!\left[
\frac{z}{D}\;\middle|\;-\,;
\right.\\[-1mm]
&\qquad\left.
\left(k-\frac{1}{\beta},\frac{1}{\beta}\right),
\left(m-\frac{1}{\kappa_{\mathrm e}},
\frac{1}{\kappa_{\mathrm e}}\right)
\right].
\end{aligned}
\label{eq:PDFgggg}
\end{equation}
Here, \((a_u,A_u)\) and \((b_v,B_v)\) denote the upper and lower
parameter pairs, respectively, in the standard Fox \(H\)-function
notation(We write \(\kappa_{\mathrm e}\) for the small-scale GG exponent to avoid clashing with the kernel symbol \(\kappa(\cdot)\).
\end{proposition}

\begin{proof}
Write \(\Gamma\!\left(k+\frac{\zeta-1}{\beta}\right)
=\Gamma\!\left(k-\frac{1}{\beta}+\frac{\zeta}{\beta}\right)\) and
\(\Gamma\!\left(m+\frac{\zeta-1}{\kappa_{\mathrm e}}\right)
=\Gamma\!\left(m-\frac{1}{\kappa_{\mathrm e}}+\frac{\zeta}{\kappa_{\mathrm e}}\right)\).
By the Mellin--Barnes definition of the Fox \(H\)-function and its
Mellin-transform identity \cite{mathai2009h}, the function
\[
\phi(z)=H^{2,0}_{0,2}\!\left[
z\;\middle|\;-\,;
\left(k-\frac{1}{\beta},\frac{1}{\beta}\right),
\left(m-\frac{1}{\kappa_{\mathrm e}},\frac{1}{\kappa_{\mathrm e}}\right)
\right]
\]
satisfies
\[
\int_0^\infty z^{\zeta-1}\phi(z)\,dz
=
\Gamma\!\left(k+\frac{\zeta-1}{\beta}\right)
\Gamma\!\left(m+\frac{\zeta-1}{\kappa_{\mathrm e}}\right).
\]
Applying the Mellin scaling property to \(\phi(z/D)/D\) and matching the
resulting Mellin transform with \eqref{eq:MZ} gives \eqref{eq:PDFgggg}.
\end{proof}

\begin{remark}
The PDF in \eqref{eq:PDFgggg} reduces to the Nakagami-\(m\)/generalized-Gamma composite model for \(\kappa_{\mathrm e}=1\), and further to the classical Nakagami-\(m\)/Gamma model for \(\kappa_{\mathrm e}=\beta=1\).
\end{remark}

\section{Joint PDF under Correlated GG Shadowing and Independent GG Small-Scale Fading}
\label{sec:jointpdf}

Only the shadowing pair \((Y_{\mathrm B},Y_{\mathrm E})\) is correlated; \(V_{\mathrm B}\) and \(V_{\mathrm E}\) remain independent of each other and of \((Y_{\mathrm B},Y_{\mathrm E})\) \cite{5281755}. The correlated-Gamma expansion can be written as a double mixture of the normalized single-link kernels in Proposition~\ref{prop:1prime}. For \(i,j\in\mathbb N_0\triangleq\{0,1,2,\ldots\}\), and using the Pochhammer symbol \((a)_n\triangleq\Gamma(a+n)/\Gamma(a)\), the coefficient inherited from the correlated-Gamma expansion for the eavesdropping link satisfies \(\displaystyle \frac{1}{(i+k_{\mathrm E})_j\,\Gamma(i+k_{\mathrm E})}=\frac{1}{\Gamma(i+j+k_{\mathrm E})}\). Hence, the factor \((i+k_{\mathrm E})_j^{-1}\) is already absorbed into the normalization of the eavesdropping-link kernel with shape \(i+j+k_{\mathrm E}\) and must not appear again in the mixture weight.

The conditional scales are \(\Theta_\ell\triangleq(1-\rho)\theta_\ell\), and \(D_\ell=\Theta_\ell^{1/\beta_\ell}\tau_\ell^{1/\kappa_\ell}\).
where \(\tau_\ell\) follows from the mean-power normalization in
\eqref{eq:tautomean}. The normalized mixture weights are defined as
\begin{equation}
w_{ij}
=
\frac{(1-\rho)^{k_{\mathrm E}}
(k_{\mathrm B})_i
(k_{\mathrm E}-k_{\mathrm B})_j
\rho^{i+j}}
{i!\,j!},
\qquad i,j\in\mathbb N_0.
\label{eq:corrected_weights}
\end{equation}
Therefore, by \cite[eq. (3)]{5281755}, for \(x_{\mathrm B},x_{\mathrm E}\geq0\),
\begin{equation}
\begin{aligned}
f_{\gamma_{\mathrm B},\gamma_{\mathrm E}}
(x_{\mathrm B},x_{\mathrm E})
&=\sum_{i=0}^{\infty}\sum_{j=0}^{\infty} w_{ij}\,
\kappa(x_{\mathrm B};i+k_{\mathrm B},m_{\mathrm B},
\beta_{\mathrm B},\kappa_{\mathrm B},D_{\mathrm B})\\
&\quad\times
\kappa(x_{\mathrm E};i+j+k_{\mathrm E},m_{\mathrm E},
\beta_{\mathrm E},\kappa_{\mathrm E},D_{\mathrm E}).
\end{aligned}
\label{eq:jointgggg}
\end{equation}

For \(k_{\mathrm E}\geq k_{\mathrm B}>0\), the weights in \eqref{eq:corrected_weights} are nonnegative and normalized, since
\begin{equation}
\begin{aligned}
\sum_{i=0}^{\infty}\sum_{j=0}^{\infty}w_{ij}
&=(1-\rho)^{k_{\mathrm E}}
\left[\sum_{i=0}^{\infty}
\frac{(k_{\mathrm B})_i}{i!}\rho^i\right]
\left[\sum_{j=0}^{\infty}
\frac{(k_{\mathrm E}-k_{\mathrm B})_j}{j!}\rho^j\right]\\
&=(1-\rho)^{k_{\mathrm E}}
(1-\rho)^{-k_{\mathrm B}}
(1-\rho)^{-(k_{\mathrm E}-k_{\mathrm B})}
=1.
\end{aligned}
\label{eq:weights_normalization}
\end{equation}
The double-series representation is therefore used under the ordering
\(k_{\mathrm E}\geq k_{\mathrm B}>0\). When \(k_{\mathrm B}>k_{\mathrm E}\), the analogous expansion obtained by interchanging the Gamma-pair indices must be used.

\section{Secrecy Performance Analysis}
\label{sec:secrecy_performance}

\subsection{Closed-Form Composite Survival Function}

\begin{proposition}\label{prop:3prime}
For \(y\geq0\), let
\[
G(y;k,m,\beta,\kappa_{\mathrm e},D)
\triangleq
\int_y^\infty
\kappa(x;k,m,\beta,\kappa_{\mathrm e},D)\,dx .
\]
Define \(\mathbf b\triangleq\left\{\left(m,\frac{1}{\kappa_{\mathrm e}}\right),\left(k,\frac{1}{\beta}\right),(0,1)\right\}\). Then, the composite survival function is
\begin{equation}
G(y;k,m,\beta,\kappa_{\mathrm e},D)
=
\frac{1}{\Gamma(k)\Gamma(m)}
H^{3,0}_{1,3}\!\left[
\frac{y}{D}\;\middle|\;(1,1)\,;\mathbf b
\right].
\label{eq:Gprime}
\end{equation}
\end{proposition}

\begin{proof}
Using the Mellin--Barnes representation of \eqref{eq:PDFgggg}, with \(\mathrm i\triangleq\sqrt{-1}\) and \(\mathcal L_\zeta\) denoting a suitable vertical Mellin--Barnes contour,
\begin{equation}
\begin{aligned}
\kappa(x;k,m,\beta,\kappa_{\mathrm e},D)
&=
\frac{1}{D\Gamma(k)\Gamma(m)}
\frac{1}{2\pi\mathrm i}
\int_{\mathcal L_\zeta}
\Gamma\!\left(k-\frac{1}{\beta}+\frac{\zeta}{\beta}\right)\\
&\quad\times
\Gamma\!\left(m-\frac{1}{\kappa_{\mathrm e}}+
\frac{\zeta}{\kappa_{\mathrm e}}\right)
\left(\frac{x}{D}\right)^{-\zeta}d\zeta .
\end{aligned}
\label{eq:pdf_MB}
\end{equation}
Choosing \(\mathcal L_\zeta\) such that \(\Re(\zeta)>1\), integrating \eqref{eq:pdf_MB} from \(y\) to \(\infty\), and then setting \(v=\zeta-1\), with \(\mathcal L_v\) denoting the shifted contour, give
\begin{equation}
\begin{aligned}
G(y;k,m,\beta,\kappa_{\mathrm e},D)
&=
\frac{1}{\Gamma(k)\Gamma(m)}
\frac{1}{2\pi\mathrm i}
\int_{\mathcal L_v}
\Gamma\!\left(k+\frac{v}{\beta}\right)\\
&\quad\times
\Gamma\!\left(m+\frac{v}{\kappa_{\mathrm e}}\right)
\frac{1}{v}
\left(\frac{y}{D}\right)^{-v}dv .
\end{aligned}
\label{eq:G_MB}
\end{equation}
Since \(1/v=\Gamma(v)/\Gamma(1+v)\), the Mellin--Barnes kernel in
\eqref{eq:G_MB} is precisely that of the Fox \(H\)-function in
\eqref{eq:Gprime}, with lower parameter pairs
\((m,1/\kappa_{\mathrm e})\), \((k,1/\beta)\), and \((0,1)\), and upper parameter pair \((1,1)\). This proves \eqref{eq:Gprime}.
\end{proof}

\subsection{Closed-Form Zero-Rate Outage / PNZSC}

We now generalize the Mellin--Barnes convolution that combines the
legitimate-link survival function with the eavesdropping-link PDF in the concrete \(Q_{ij}\) form.

For a single \((i,j)\) mixture term, let \(I\) and \(J\) denote the corresponding latent mixture indices, and define \(s_{\mathrm B}\triangleq i+k_{\mathrm B}\), \(s_{\mathrm E}\triangleq i+j+k_{\mathrm E}\), \(D_{\mathrm B}\triangleq\Theta_{\mathrm B}^{1/\beta_{\mathrm B}}\tau_{\mathrm B}^{1/\kappa_{\mathrm B}}\), and \(D_{\mathrm E}\triangleq\Theta_{\mathrm E}^{1/\beta_{\mathrm E}}\tau_{\mathrm E}^{1/\kappa_{\mathrm E}}\). Also, let \(G_{\mathrm B,ij}(x)\triangleq G(x;s_{\mathrm B},m_{\mathrm B},\beta_{\mathrm B},\kappa_{\mathrm B},D_{\mathrm B})\) and \(f_{\mathrm E,ij}(x)\triangleq\kappa(x;s_{\mathrm E},m_{\mathrm E},\beta_{\mathrm E},\kappa_{\mathrm E},D_{\mathrm E})\). Then,
\begin{equation}
\begin{aligned}
Q_{ij}
&\triangleq
\Pr[\gamma_{\mathrm B}>\gamma_{\mathrm E}\mid I=i,J=j]\\
&=\int_0^\infty
G_{\mathrm B,ij}(x)\,
f_{\mathrm E,ij}(x)\,dx .
\end{aligned}
\label{eq:Qijdef}
\end{equation}

\begin{proposition}\label{prop:4prime}
For the complex Mellin variable \(q\) satisfying \(\Re(q)>0\), the Mellin transform of the legitimate-link survival function is
\begin{equation}
\mathcal M_{G_{\mathrm B,ij}}(q)
=
\frac{D_{\mathrm B}^q}{q}\,
\frac{\Gamma\!\left(m_{\mathrm B}+\frac{q}{\kappa_{\mathrm B}}\right)
\Gamma\!\left(s_{\mathrm B}+\frac{q}{\beta_{\mathrm B}}\right)}
{\Gamma(m_{\mathrm B})\Gamma(s_{\mathrm B})}.
\label{eq:MG1prime}
\end{equation}
For \(\Re(q)>1-\min\{\kappa_{\mathrm E}m_{\mathrm E},
\beta_{\mathrm E}s_{\mathrm E}\}\), the Mellin transform of the eavesdropping-link PDF is
\begin{equation}
\mathcal M_{f_{\mathrm E,ij}}(q)
=
D_{\mathrm E}^{q-1}\,
\frac{\Gamma\!\left(m_{\mathrm E}+\frac{q-1}{\kappa_{\mathrm E}}\right)
\Gamma\!\left(s_{\mathrm E}+\frac{q-1}{\beta_{\mathrm E}}\right)}
{\Gamma(m_{\mathrm E})\Gamma(s_{\mathrm E})}.
\label{eq:Mkappa2prime}
\end{equation}
Using Mellin--Parseval's identity over a vertical contour \(\mathcal L\) satisfying \(0<\Re(q)<\min\{\kappa_{\mathrm E}m_{\mathrm E},\beta_{\mathrm E}s_{\mathrm E}\}\),
\[
Q_{ij}
=
\frac{1}{2\pi\mathrm i}
\int_{\mathcal L}
\mathcal M_{G_{\mathrm B,ij}}(q)\,
\mathcal M_{f_{\mathrm E,ij}}(1-q)\,dq .
\]
Define \(\mathbf a_{ij}\triangleq\left\{\left(1-m_{\mathrm E},\frac{1}{\kappa_{\mathrm E}}\right),\left(1-s_{\mathrm E},\frac{1}{\beta_{\mathrm E}}\right),(1,1)\right\}\) and \(\mathbf b_{ij}\triangleq\left\{\left(m_{\mathrm B},\frac{1}{\kappa_{\mathrm B}}\right),\left(s_{\mathrm B},\frac{1}{\beta_{\mathrm B}}\right),(0,1)\right\}\).
Then,
\begin{equation}
Q_{ij}
=
\frac{1}
{\Gamma(s_{\mathrm B})\Gamma(m_{\mathrm B})
\Gamma(s_{\mathrm E})\Gamma(m_{\mathrm E})}
H_{3,3}^{3,2}\!\left[
\frac{D_{\mathrm E}}{D_{\mathrm B}}
\;\middle|\;
\mathbf a_{ij}\,;\mathbf b_{ij}
\right].
\label{eq:Qijprime}
\end{equation}
Consequently,
\begin{equation}
\mathrm{PNZSC}
=
\sum_{i=0}^{\infty}\sum_{j=0}^{\infty}w_{ij}Q_{ij},
\qquad
P_o(0)=1-\mathrm{PNZSC}.
\label{eq:PNZSCprime}
\end{equation}
\end{proposition}

\begin{proof}
The Mellin transform in \eqref{eq:MG1prime} follows from the general tail identity
\[
\int_0^\infty x^{q-1}\Pr[Z>x]\,dx
=
\frac{\mathbb E[Z^q]}{q},
\qquad \Re(q)>0,
\]
together with the \(q\)-th moment obtained from \eqref{eq:MZ}. Equation~\eqref{eq:Mkappa2prime} follows directly from \eqref{eq:MZ} after replacing \(k\), \(m\), \(\beta\), \(\kappa\), and \(D\) by \(s_{\mathrm E}\), \(m_{\mathrm E}\), \(\beta_{\mathrm E}\), \(\kappa_{\mathrm E}\), and \(D_{\mathrm E}\), respectively.

Substituting \eqref{eq:MG1prime} and \eqref{eq:Mkappa2prime} into the Parseval integral gives
\begin{equation}
\begin{aligned}
&Q_{ij}\Gamma(s_{\mathrm B})\Gamma(m_{\mathrm B})
\Gamma(s_{\mathrm E})\Gamma(m_{\mathrm E})\\
&=
\frac{1}{2\pi\mathrm i}
\int_{\mathcal L}
\frac{
\Gamma\!\left(m_{\mathrm B}+\frac{q}{\kappa_{\mathrm B}}\right)
\Gamma\!\left(s_{\mathrm B}+\frac{q}{\beta_{\mathrm B}}\right)}
{q}\\
&\quad\times
\Gamma\!\left(m_{\mathrm E}-\frac{q}{\kappa_{\mathrm E}}\right)
\Gamma\!\left(s_{\mathrm E}-\frac{q}{\beta_{\mathrm E}}\right)
\left(\frac{D_{\mathrm B}}{D_{\mathrm E}}\right)^q dq .
\end{aligned}
\label{eq:Qij_MB_integral}
\end{equation}
Using \(\frac{1}{q}=\frac{\Gamma(q)}{\Gamma(1+q)}\) and \(\left(\frac{D_{\mathrm B}}{D_{\mathrm E}}\right)^q=\left(\frac{D_{\mathrm E}}{D_{\mathrm B}}\right)^{-q}\), the factors in \eqref{eq:Qij_MB_integral} are identified with the standard Mellin--Barnes definition of the Fox \(H\)-function as follows. By comparison with the standard Mellin--Barnes representation, the Fox-\(H\) orders are \(p_{\mathrm H}=q_{\mathrm H}=m_{\mathrm H}=3\) and \(n_{\mathrm H}=2\), with
\[
\begin{gathered}
(a_1,A_1)=(1-m_{\mathrm E},\kappa_{\mathrm E}^{-1}),\quad
(a_2,A_2)=(1-s_{\mathrm E},\beta_{\mathrm E}^{-1}),\\
(a_3,A_3)=(1,1),\quad
(b_1,B_1)=(m_{\mathrm B},\kappa_{\mathrm B}^{-1}),\\
(b_2,B_2)=(s_{\mathrm B},\beta_{\mathrm B}^{-1}),\quad
(b_3,B_3)=(0,1).
\end{gathered}
\]
Collecting two reflected upper-numerator terms, three direct lower-numerator terms, one direct upper-denominator term, and no reflected lower-denominator terms yields the Fox \(H\)-function in \eqref{eq:Qijprime}. Summing over \((i,j)\) with the normalized weights in \eqref{eq:corrected_weights} gives \eqref{eq:PNZSCprime}.
\end{proof}

\begin{figure}[t]
    \centering
    \includegraphics[width=\linewidth]{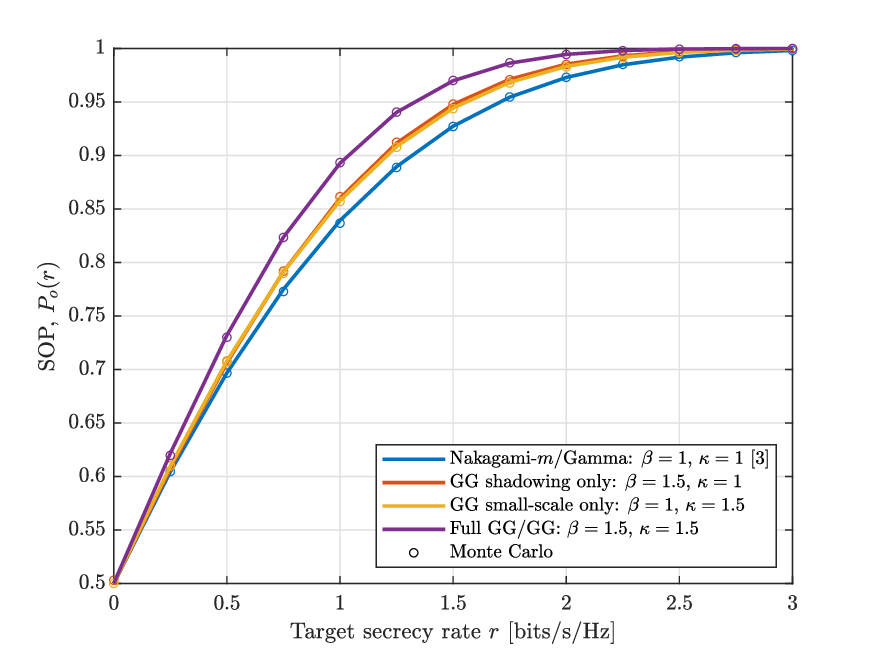}
    \caption{SOP-vs-target-rate plot for the GG/GG fading model}
    \label{fig:Fig_1}
\end{figure}

\begin{figure*}[t]
    \centering
    \includegraphics[width=\linewidth]{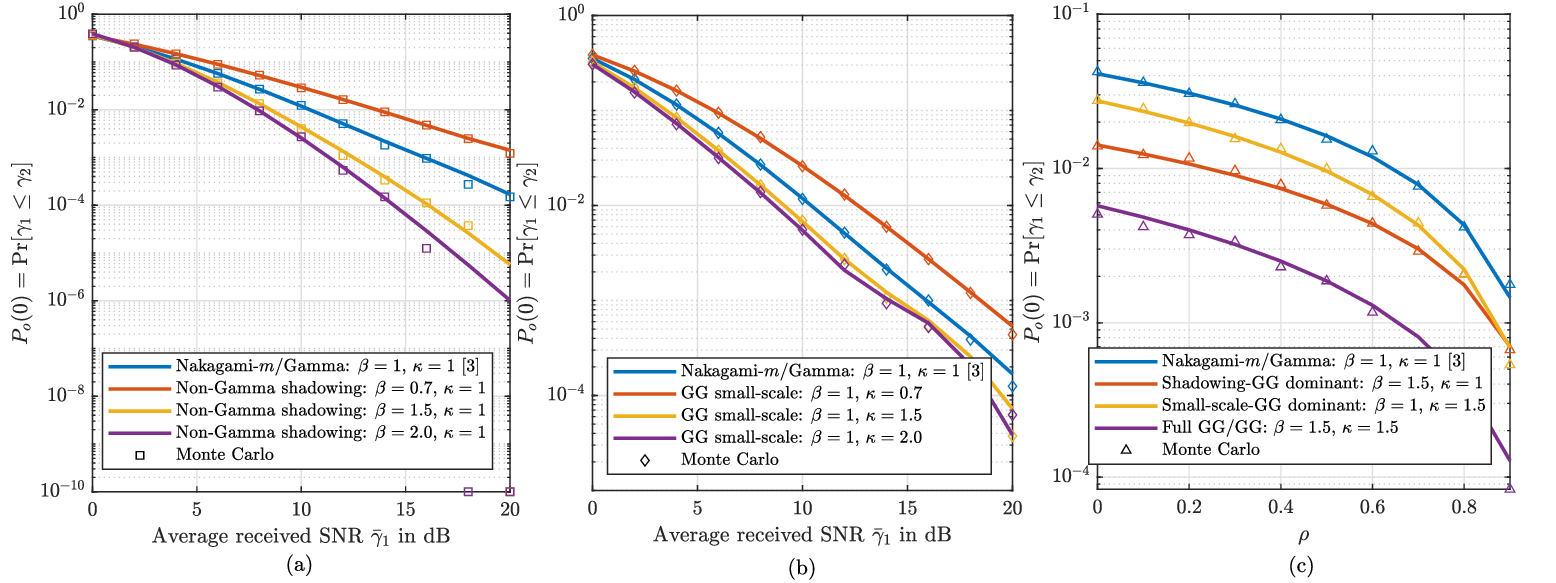}
    \caption{Effects of the GG/GG fading parameters on the zero-rate secrecy outage probability $P_o(0)$: (a) impact of the shadowing exponent $\beta$; (b) impact of the small-scale fading parameter $\kappa$ for fixed $\beta=1$; and (c) impact of the correlation coefficient $\rho$ of the underlying Gamma pair.}
    \label{fig:Fig_2}
\end{figure*}

\subsection{Secrecy Outage Probability at a General Rate}

For $r>0$, define $h(x_2,r)\triangleq 2^r(1+x_2)-1$. For a fixed \((i,j)\), the integral of the link-1 survival function over the link-2 PDF is the secrecy-success probability conditioned on \((i,j)\). Therefore,
\begin{equation}
\begin{aligned}
P_o(r)
&=1-
\sum_{i=0}^{\infty}\sum_{j=0}^{\infty}w_{ij}
\int_0^\infty
G\big(h(x_2,r);i+k_1,m_1,\beta_1,\\ 
&\quad\kappa_1,D_1\big)\times\kappa(x_2;i+j+k_2,m_2,\beta_2,\kappa_2,D_2)\,dx_2 .
\end{aligned}
\label{eq:SOPgen}
\end{equation}
At \(r=0\), \(h(x_2,0)=x_2\), and the double-series integral in
\eqref{eq:SOPgen} reduces to
\(\sum_{i,j}w_{ij}Q_{ij}=\mathrm{PNZSC}\). Hence,
\eqref{eq:SOPgen} consistently gives
\(P_o(0)=1-\mathrm{PNZSC}\).

The convolution technique of Proposition~\ref{prop:4prime} cannot be
applied directly for \(r>0\) because \(h(x_2,r)\) is affine rather than purely multiplicative in \(x_2\). The Mellin--Barnes convolution identity requires the two functions to share the integration variable up to a multiplicative scale, which is prevented by the additive term \(2^r-1\) in \(h(x_2,r)\).

For the Nakagami-\(m\)/Gamma baseline, corresponding to
\(\beta_\ell=\kappa_\ell=1\), \cite{8060499} exploits the Bessel-\(K\)
representation and Kummer's transformation to construct an efficient
numerical quadrature. This specific reduction does not directly extend to
the general GG/GG kernel. Therefore, \eqref{eq:SOPgen} is retained as an
exact double-series representation with one residual one-dimensional
integral per \((i,j)\) term, which can be evaluated numerically using the
Fox-\(H\) representations of the survival function and PDF.

\section{Numerical Results and Discussion}
\label{sec:numerical_results}

This section evaluates the secrecy outage performance of the proposed correlated composite GG/GG fading model. Fig.~\ref{fig:Fig_1} depicts $P_o(r)$ versus the target secrecy rate $r$ for $k_1=k_2=1$, $m_1=m_2=4$, $\rho=0.6$, and $\bar{\gamma}_1=\bar{\gamma}_2=4$ dB. The outage probability increases with $r$ for all cases, since a higher target secrecy rate requires a stronger legitimate channel advantage. Because the legitimate and eavesdropping links have the same average SNR, the outage probability is already high at low rates and approaches unity as $r$ increases. The difference between the baseline, the GG-shadowing-only case, the GG-small-scale-only case, and the full GG/GG case shows that both $\beta$ and $\kappa$ significantly affect secrecy performance.

Fig.~\ref{fig:Fig_2}(a)-(c) demonstrates the effects of $\beta$, $\kappa$, and $\rho$ on the zero-rate outage probability $P_o(0)$. In Fig.~\ref{fig:Fig_2}(a), different values of the shadowing exponent $\beta$ lead to clearly different outage levels, confirming the importance of generalized shadowing. In Fig.~\ref{fig:Fig_2}(b), with $\beta=1$ fixed, varying $\kappa$ changes both the outage level and its decay with increasing $\bar{\gamma}_1$, showing the impact of generalized small-scale fading. In Fig.~\ref{fig:Fig_2}(c), $P_o(0)$ varies noticeably with the correlation parameter $\rho$ of the underlying Gamma pair, highlighting the role of correlated shadowing between the legitimate and eavesdropping links. 

Overall, Figs.~\ref{fig:Fig_1} and~\ref{fig:Fig_2} show that the proposed GG/GG model provides a more flexible secrecy outage characterization than the Nakagami-$m$/Gamma baseline. The parameter $\beta$ captures generalized shadowing effects, while $\kappa$ captures generalized small-scale fading effects. Their joint variation enables the model to describe a wider range of physical-layer security behaviors under correlated composite fading.

\section{Conclusion}
\label{sec:conclusion}

This paper developed a unified secrecy-outage framework for correlated composite generalized-Gamma/generalized-Gamma fading channels. Closed-form expressions for the single-link PDF, joint distribution, survival function, and zero-rate SOP/PNZSC were derived using Mellin transforms and Fox-$H$ functions, while the general-rate SOP was represented by an exact double series with one residual integral per term. The model includes the Nakagami-$m$/generalized-Gamma and Nakagami-$m$/Gamma cases as special cases. Numerical and Monte Carlo results confirmed the analysis and showed the significant effects of the fading exponents and shadowing correlation on secrecy performance.

\bibliographystyle{IEEEtran}
\bibliography{References.bib}

@ARTICLE{11036803,
  author={Zhang, Xiaoqi and Ma, Xu and Zhang, Haijun and Ren, Yuzheng and Cao, Difei and Leung, Victor C. M.},
  journal={IEEE Communications Letters}, 
  title={Physical Layer Security Enhancement for Beyond Diagonal {RIS-Aided MIMO} Communications}, 
  year={2025},
  volume={29},
  number={8},
  pages={1919-1923},
  doi={10.1109/LCOMM.2025.3579860}}

@ARTICLE{8060499,
  author={Alexandropoulos, George C. and Peppas, Kostas P.},
  journal={IEEE Communications Letters}, 
  title={Secrecy Outage Analysis Over Correlated Composite {Nakagami-$m$} /Gamma Fading Channels}, 
  year={2018},
  volume={22},
  number={1},
  pages={77-80},
  doi={10.1109/LCOMM.2017.2760255}}

@book{mathai2009h,
  title={The H-function: theory and applications},
  author={Mathai, Arakaparampil M and Saxena, Ram Kishore and Haubold, Hans J},
  year={2009},
  publisher={Springer Science \& Business Media}
}

@ARTICLE{6772207,
  author={Wyner, A. D.},
  journal={The Bell System Technical Journal}, 
  title={The wire-tap channel}, 
  year={1975},
  volume={54},
  number={8},
  pages={1355-1387},
  doi={10.1002/j.1538-7305.1975.tb02040.x}}

@ARTICLE{9239955,
  author={Badarneh, Osamah S. and Sofotasios, Paschalis C. and Muhaidat, Sami and Cotton, Simon L. and Rabie, Khaled M. and Aldhahir, Naofal},
  journal={IEEE Access}, 
  title={Achievable Physical-Layer Security Over Composite Fading Channels}, 
  year={2020},
  volume={8},
  number={},
  pages={195772-195787},
  doi={10.1109/ACCESS.2020.3033893}}

@INPROCEEDINGS{9768319,
  author={Eldokmak, Youssef M. and Ismail, Mahmoud H. and Hassan, Mohamed S.},
  booktitle={2022 Wireless Telecommunications Symposium (WTS)}, 
  title={Physical Layer Security Analysis Over Composite Generalized Gamma-Lognormal Fading Channels}, 
  year={2022},
  volume={},
  number={},
  pages={1-6},
  doi={10.1109/WTS53620.2022.9768319}}

@ARTICLE{8917797,
  author={N. Kamga, Gervais and Aïssa, Sonia and Rasethuntsa, Tau R. and Alouini, Mohamed-Slim},
  journal={IEEE Transactions on Wireless Communications}, 
  title={Mixed {RF/FSO} Communications With Outdated-{CSI}-Based Relay Selection Under Double Generalized Gamma Turbulence, Generalized Pointing Errors, and {Nakagami-$m$} Fading}, 
  year={2021},
  volume={20},
  number={5},
  pages={2761-2775},
  doi={10.1109/TWC.2019.2954866}}

@ARTICLE{5281755,
  author={Bithas, Petros S. and Sagias, Nikos C. and Mathiopoulos, P. Takis},
  journal={IEEE Transactions on Communications}, 
  title={The bivariate generalized-{$K$} ({$K_G$}) distribution and its application to diversity receivers}, 
  year={2009},
  volume={57},
  number={9},
  pages={2655-2662},
  doi={10.1109/TCOMM.2009.09.080039}}

\end{document}